\documentclass[superscriptaddress]{revtex4}
\usepackage{braket}
\usepackage{amsmath}
\usepackage{amssymb}
\usepackage{amsthm}
\usepackage{graphicx}
\usepackage[colorlinks,
            linkcolor=blue, 
            anchorcolor=blue,
            citecolor=blue]{hyperref}
\usepackage{appendix}
\theoremstyle{definition}
\newtheorem{theorem}{Theorem}

\newtheorem{lemma}{Lemma}
\newtheorem{corollary}{Corollary}
\newtheorem*{example}{Example}

\newtheorem*{conjecture}{Conjecture}

\begin{document}
\title{Estimating the concurrence for quantum states via symmetric measurements}
\author{Hao-Fan Wang}
\affiliation{LMIB(Beihang University), Ministry of Education, and School of Mathematical Sciences, Beihang University, Beijing 100191, China}
\author{Wen Zhou}
\affiliation{School of Mathematics, Liaoning Normal University, Dalian 116029, China}
\author{Lin Chen}
\email{linchen@buaa.edu.cn(corresponding author)}
\affiliation{LMIB(Beihang University), Ministry of Education, and School of Mathematical Sciences, Beihang University, Beijing 100191, China}
\author{Shao-Ming Fei}
\email{feishm@cnu.edu.cn}
\affiliation{School of Mathematical Sciences, Capital Normal University, Beijing 100048, China}
\begin{abstract}
Concurrence is one of the most widely used entanglement measures. We show that, for arbitrary bipartite states, the lower bound of concurrence induced by symmetric measurements is always tighter than the one obtained from realignment. Moreover, we derive an extended version of the lower bound of concurrence induced by symmetric measurements, which can also be experimentally accessed without state tomography, and further present a corresponding lower bound of the genuine tripartite entanglement concurrence.
\end{abstract}

\maketitle

\section{Introduction}

Quantifying and estimating the entanglement of composite systems without reconstructing the entire density matrix is a problem of considerable importance in quantum information theory \cite{Tomography}. Full quantum state tomography provides complete information about a state, but its measurement cost increases sharply with the system size. For two-qubit systems, concurrence \cite{Rungta} is a notable exception in that it admits an analytic formula for mixed states \cite{Wootters}. Indeed, concurrence is already one of the most well-known entanglement measures, and it is closely related to the entanglement of formation \cite{Hill}. However, in high-dimensional systems, owing to the infimum involved in the calculation, analytical formulas for concurrence are obtained only for some very specific quantum states \cite{isoconcurrence,Wernerconcurrence}. \par

In place of the analytical formula for concurrence, much progress has been made on lower bounds of concurrence for general bipartite states. In 2005, a lower bound of concurrence \cite{LowBound0} was proposed based on the positive partial transpose (PPT) criterion \cite{PPT1,PPT2} and the realignment criterion \cite{CCNR1,CCNR2}. Subsequently, the lower bound of concurrence based on local uncertainty relations and correlation matrix was introduced \cite{LB1}. Since then various lower bounds have been derived using various methods \cite{LB2,LB3,LB4,LB5,LB6,LB7}, including the ones based on the extended realigned matrix \cite{LBER1,LBER2}. Recently, the lower bound of concurrence based on symmetric measurements was proposed in Ref.\cite{HFWang1} by employing the conical 2-design property. This bound can be experimentally accessed without state tomography, and it has been generalized to versions of generalized symmetric measurements \cite{LBGSM} and generalized equiangular measurements \cite{LBGEM} by using a similar method.\par

Symmetric informationally complete positive operator-valued measure (SIC POVM) \cite{SIC} and complete set of mutually unbiased bases (MUBs) \cite{MUB} are the two most famous examples of complex projective 2-designs, which are further generalized to general SIC POVM (GSIC POVM) \cite{GSIC} and mutually unbiased measurements (MUMs) \cite{MUM}. Graydon and Appleby proposed the notion of conical 2-designs \cite{Graydon1} as a generalization of complex projective 2-designs. Both GSIC POVM and informationally complete MUMs are the instances of this notion \cite{Graydon2,KWang}. The informationally complete $(N,M)$-POVM, also called symmetric measurements \cite{SM0}, is an important example of conical 2-designs \cite{SMdesign1,SMdesign2}, which reduces to the GSIC POVM and MUMs in special cases. It has been shown in Refs.\cite{HFWang1,HFWang2,LBGSM} that some conical 2-designs can be used to derive lower bounds of concurrence and Schmidt number criteria. In addition, uncertainty relations for quantum measurements assigned to conical 2-designs have been proposed in Refs.\cite{Uncertainty1,Uncertainty2}.\par

In this work, we first show that, for arbitrary bipartite states, the lower bound of concurrence induced by symmetric measurements is always tighter than the one obtained from realignment. In other words, we prove the Conjecture posed in Ref.\cite{HFWang1}, which means that some results obtained from symmetric measurements are strictly stronger than the ones obtained from realignment. Furthermore, we also propose several open problems about lower bounds of concurrence induced by conical 2-designs. In addition, we derive an extended version of the lower bound of concurrence induced by symmetric measurements, which can likewise be experimentally accessed without state tomography, and further present a corresponding lower bound of genuine tripartite entanglement concurrence.\par

The remainder of this paper is organized as follows. In Sec.~\ref{sec:pre}, we clarify the fundamental concepts and main knowledge used in this work. In Sec.~\ref{sec3}, we prove the Conjecture in Ref.\cite{HFWang1} is true. In Sec.~\ref{sec4}, we derive an improved lower bound of concurrence induced by symmetric measurements and a corresponding lower bound of genuine tripartite entanglement concurrence. Finally, in Sec.~\ref{sec5}, we summarize our results and provide discussion.\par

\section{Preliminaries}\label{sec:pre}
In this section, we introduce the main facts and tools used in this paper. First, we recall the definitions of conical 2-designs and symmetric measurements, as well as some of their properties. A conical 2-design \cite{Graydon1,Graydon2,KWang} is a finite set $E=\{E_{k}\}_{k=1}^{m}$ of positive semidefinite operators on the $d$-dimensional Hilbert space $\mathcal{H}$ satisfying
    \begin{equation*}
        \sum\limits_{k=1}^{m}E_{k}\otimes E_{k}=\alpha_{E}I+\beta_{E}F, 
    \end{equation*}
for some constants $\alpha_{E}\geq\beta_{E}>0$, where $I$ and $F$ are the identity operator and the swap operator on $\mathcal{H}\otimes\mathcal{H}$.\par

The informationally complete $(N,M)$-POVM is a prominent example of conical 2-designs. An $(N,M)$-POVM \cite{SM0} is a collection of $N$ $d$-dimensional POVMs $\{E_{\alpha,k}|k=1,2\cdots,M\}$ ($\alpha=1,2,\cdots,N$) satisfying the following conditions,
\begin{equation}\label{eq:(N,M)-POVM}
    \begin{split}
    {\rm tr}(E_{\alpha,k}) &= \dfrac{d}{M}, \\
	{\rm tr}(E_{\alpha,k}^{2}) &= x,\\
	{\rm tr}(E_{\alpha,k}E_{\alpha,l}) &= \dfrac{d-Mx}{M(M-1)},~~ l\neq k\\
	{\rm tr}(E_{\alpha,k}E_{\beta,l}) &= \dfrac{d}{M^{2}},~~ \beta\neq\alpha
    \end{split}
\end{equation}
where the parameter $x$ satisfies $\dfrac{d}{M^{2}}<x\leq \min\left\{\dfrac{d^{2}}{M^{2}},\dfrac{d}{M}\right\}$. When $N(M-1)=d^{2}-1$, the $(N,M)$-POVM is called an informationally complete $(N,M)$-POVM. For any finite dimension $d$ ($d>2$), there exist at least four distinct classes of informationally complete $(N,M)$-POVM: (1) $N=1$ and $M=d^{2}$ (GSIC POVM), (2) $N=d+1$ and $M=d$ (MUMs), (3) $N=d^{2}-1$ and $M=2$, (4) $N=d-1$ and $M=d+2$. Indeed, informationally complete $(N,M)$-POVMs have the following property.\par

\begin{lemma}\label{NMconi2_Lemma}
    \cite{SMdesign1,SMdesign2} Let $\{E_{\alpha,k}|\alpha=1,2,\cdots,N;k=1,2,\cdots,M\}$ be an informationally complete $(N,M)$-POVM with free parameter $x$ on the $d$-dimensional Hilbert space $\mathcal{H}$. Then
    \begin{equation*}
        \sum\limits_{\alpha=1}^{N}\sum\limits_{k=1}^{M}E_{\alpha,k}\otimes E_{\alpha,k}=\frac{d^{3}-xM^{2}}{dM(M-1)}I+\frac{xM^{2}-d}{M(M-1)}F,
    \end{equation*}
    where $I$ and $F$ are the identity operator and the swap operator on $\mathcal{H}\otimes\mathcal{H}$. This lemma shows that every informationally complete $(N,M)$-POVM is a conical 2-design. 
\end{lemma}

From orthonormal Hermitian operator basis $\{G_{0}=\mathbb{I}_{d}/\sqrt{d},\,G_{\alpha,k}|\alpha=1,\cdots,N;\,k=1,\cdots,M-1\}$ with ${\rm tr}(G_{\alpha,k})=0$, an informationally complete $(N,M)$-POVM is given by
\begin{equation}
	E_{\alpha,k}=\dfrac{1}{M}\mathbb{I}_{d}+tH_{\alpha,k},
\end{equation}
where
\begin{equation}
	H_{\alpha,k}=\begin{cases}
		G_{\alpha}-\sqrt{M}(\sqrt{M}+1)G_{\alpha,k},~ k=1,\cdots,M-1\\
		(\sqrt{M}+1)G_{\alpha}, k=M
	\end{cases}
\end{equation}
with $G_{\alpha}=\sum\limits_{k=1}^{M-1}G_{\alpha,k}$. The parameter $t$ should be chosen such that $E_{\alpha,k}\geq 0$, which is equivalent to
\begin{equation}
	-\dfrac{1}{M}\dfrac{1}{\lambda_{\max}}\leq t \leq \dfrac{1}{M}\dfrac{1}{|\lambda_{\min}|},
\end{equation}
where $\lambda_{\max}$ and $\lambda_{\min}$ are the maximal and minimal eigenvalue from among all eigenvalues of $H_{\alpha,k}$, respectively. (Since ${\rm tr}(H_{\alpha,k})=0$, these operators must have both negative and positive eigenvalues.) The parameters $t$ and $x$ satisfy the following relation,
\begin{equation}
    x=\dfrac{d}{M^{2}}+t^{2}(M-1)(\sqrt{M}+1)^{2}.
\end{equation}

Next, we recall the definition and some properties of concurrence. Let $\mathcal{H}_{A}$ and $\mathcal{H}_{B}$ be Hilbert spaces of dimensions $d_{A}$ and $d_{B}$, respectively. The concurrence of a pure state $\ket{\psi}\in\mathcal{H}_{A}\otimes\mathcal{H}_{B}$ is defined by \cite{Rungta}
\begin{equation}
    C(\ket{\psi})=\sqrt{2[1-{\rm tr}(\rho_{A}^{2})]}, 
\end{equation}
where $\rho_{A}={\rm tr}_{B}(\ket{\psi}\bra{\psi})$ is the reduced state obtained by the partial trace computed over subsystem $B$. The concurrence is extended to mixed states $\rho$ by the convex roof extension \cite{LowBound0}, 
\begin{equation}
    C(\rho)=\min\limits_{\{p_{i},\ket{\psi_{i}}\}}\sum_{i}p_{i}C(\ket{\psi_{i}}),
\end{equation}
where the minimum is taken over all possible pure state decompositions of $\rho=\sum\limits_{i}p_{i}\ket{\psi_{i}}\bra{\psi_{i}}$, with $p_{i}\geq 0$ and $\sum\limits_{i}p_{i}=1$. The lower bound of concurrence obtained from realignment \cite{LowBound0} can be stated as follows. For any bipartite state $\rho$ acting on $\mathcal{H}_{A}\otimes\mathcal{H}_{B}$, one has 
\begin{equation}\label{conRE}
C(\rho)\geq \sqrt{\dfrac{2}{d(d-1)}}(\|\mathcal{R}(\rho)\|_{\rm tr}-1),
\end{equation}
where $d=\min\{d_{A},d_{B}\}$, $\mathcal{R}(\rho)=\sum\limits_{i,j=1}^{d_{A}}\sum\limits_{k,l=1}^{d_{B}}\braket{ik|\rho|jl}\ket{ij}\bra{kl}$ is the realigned matrix of $\rho$, and $\|X\|_{\rm tr}={\rm tr}\sqrt{X^{\dagger}X}$ is the trace norm.\par

In addition, informationally complete symmetric measurements give rise to the following lower bound of concurrence.\par

\begin{lemma}\label{NMLOWBOUND_Lemma}
    \cite{HFWang1,HFWang2,LBGSM} Let $\{E_{\alpha,k}^{A}|\alpha=1,\cdots,N_{A};\,k=1,\cdots,M_{A}\}$ be an informationally complete $(N_{A}, M_{A})$-POVM on the $d_{A}$-dimensional Hilbert space $\mathcal{H}_{A}$, and let $\{E_{\beta,l}^{B}|\beta=1,\cdots,N_{B};\,l=1,\cdots,M_{B}\}$ be an informationally complete $(N_{B}, M_{B})$-POVM on the $d_{B}$-dimensional Hilbert space $\mathcal{H}_{B}$. Denote $p_{\alpha,k;\beta,l}={\rm tr}\left(\rho_{AB}\left(E_{\alpha,k}^{A}\otimes E_{\beta,l}^{B}\right)\right)$ and $\mathcal{P}(\rho_{AB})=\left(p_{\alpha,k;\beta,l}\right)_{N_{A}M_{A}\times N_{B}M_{B}}$ for a bipartite state $\rho_{AB}$ acting on $\mathcal{H}_{A}\otimes\mathcal{H}_{B}$. Then concurrence $C(\rho_{AB})$ satisfies 
    \begin{equation}\label{NMLOWBOUND_inequality}
        C(\rho_{AB})\geq K\sqrt{\frac{2}{d(d-1)}}\left(\|\mathcal{P}(\rho_{AB})\|_{\rm tr}-
    \sqrt{\dfrac{(d_{A}-1)(x_{A}M_{A}^{2}+d_{A}^{2})}{d_{A}M_{A}(M_{A}-1)}}\sqrt{\dfrac{(d_{B}-1)(x_{B}M_{B}^{2}+d_{B}^{2})}{d_{B}M_{B}(M_{B}-1)}}\right),
    \end{equation}
    where $K=\sqrt{\frac{M_{A}(M_{A}-1)M_{B}(M_{B}-1)}{(x_{A}M_{A}^{2}-d_{A})(x_{B}M_{B}^{2}-d_{B})}}$ and $d=\min\{d_{A},d_{B}\}$.
\end{lemma}

When $d_{A}=d_{B}$, $(N_{A},M_{A})=(N_{B},M_{B})=(N,M)$ and $x_{A}=x_{B}=x$,  (\ref{NMLOWBOUND_inequality}) reduces to the following inequality. 
\begin{equation}\label{inequality_conSM}
C(\rho)\geq \dfrac{M(M-1)}{xM^{2}-d}\sqrt{\dfrac{2}{d(d-1)}}\left(\|\mathcal{P}(\rho)\|_{\rm tr}-\dfrac{(d-1)(xM^{2}+d^{2})}{dM(M-1)}\right).
\end{equation}\par

In fact, Eq.\eqref{inequality_conSM} is a special case of Theorem 6 in Ref.\cite{LBGEM}.\par

Following the proof method of Eq.\eqref{inequality_conSM} presented in Ref.\cite{HFWang1}, we can obtain the lower bounds of concurrence induced by conical 2-designs as follows. Let $E=\{E_{k}\}_{k=1}^{m}$ be a conical 2-design on the $d$-dimensional Hilbert space $\mathcal{H}$. For any bipartite state $\rho$ acting on $\mathcal{H}\otimes\mathcal{H}$, we define $\mathcal{P}_{E}(\rho)=\sum\limits_{k=1}^{m}\sum\limits_{l=1}^{m}\operatorname{tr}\bigl[\rho(E_{k}\otimes E_{l})\bigr]\ket{w_k}\bra{w_{l}}$, where $\{\ket{w_{k}}|k=1,2,\cdots,m\}$ are orthonormal basis for $\mathbb{C}^{m}$. Then concurrence $C(\rho)$ satisfies
\begin{equation*}
    C(\rho)\geq B_{E}(\rho):=\beta_{E}^{-1}\sqrt{\frac{2}{d(d-1)}}\left(\|\mathcal{P}_{E}(\rho)\|_{\rm tr}-\beta_{E}-\alpha_{E}\right).
\end{equation*}
We provide a proof of the inequality $C(\rho)\geq B_E(\rho)$ in Appendix \ref{Appendix B_{E}}.\par

Finally, we recall the definition and basic properties of partial isometries. Let $\mathcal{H}$ and $\mathcal{K}$ be Hilbert spaces and $\mathcal{B}(\mathcal{H},\mathcal{K})$ be the set of all bounded linear operators defined on $\mathcal{H}$ with values in $\mathcal{K}$. An operator $U\in\mathcal{B}(\mathcal{H},\mathcal{K})$ is a partial isometry in case $U^{\dagger}U$ is an orthogonal projection in $\mathcal{H}$. If $U$ is a partial isometry, then the final space of $U$ is ${\rm ran}(U)\subseteq\mathcal{K}$, and the initial space of $U$ is $\ker(U)^{\perp}\subseteq\mathcal{H}$. A fundamental characterization of partial isometries is the following.\par

\begin{lemma}\label{lemma_pariso}
    (Lemma 2.30 in Ref.\cite{Carlen}) Let $\mathcal{H}$ and $\mathcal{K}$ be finite-dimensional Hilbert spaces. An operator $U\in\mathcal{B}(\mathcal{H},\mathcal{K})$ is a partial isometry if and only if there exist orthonormal sets $\{\ket{v_{1}},\cdots,\ket{v_{r}}\}$ in $\mathcal{K}$ and $\{\ket{u_{1}},\cdots,\ket{u_{r}}\}$ in $\mathcal{H}$ such that
    \begin{equation*}
        U=\sum\limits_{j=1}^{r}\ket{v_{j}}\bra{u_{j}}.
    \end{equation*}
    In this case, the initial space of $U$ is ${\rm span}(\{\ket{u_{1}},\cdots,\ket{u_{r}}\})$, and the final space of $U$ is ${\rm span}(\{\ket{v_{1}},\cdots,\ket{v_{r}}\})$.
\end{lemma}

We will use Lemma \ref{lemma_pariso} in the proof of Theorem \ref{Result2}.

\section{Lower bounds from symmetric measurements are always tighter than those from realignment}\label{sec3}
In this section, our goal is to prove the conjecture presented in Ref.\cite{HFWang1} : although the lower bound (\ref{inequality_conSM}) is intended for unknown states, it is always tighter than lower bound (\ref{conRE}) for known states. In fact, this conjecture is stronger than Conjecture 1 in Ref.\cite{JWShang}. \par

Firstly, we need the following two lemmas. 
\begin{lemma}\label{lemma1}
    (Lemma 10.2 in Ref.\cite{Jivulescu})
    Let $\{\ket{a_1},\cdots,\ket{a_n}\}$ and $\{\ket{b_1},\cdots,\ket{b_n}\}$ be two orthonormal bases of $\mathbb{C}^{n}$. For complex numbers $\gamma_{1},\cdots,\gamma_{n}$ with $|\gamma_{i}|\geq 1$ for all $1\leq i\leq n$, define the matrix 
    \begin{equation}\label{key}
        \mathcal{E}:=\sum\limits_{i=1}^{n}\gamma_{i}\ket{a_{i}}\bra{b_{i}}.
    \end{equation}
    Then, for any $X\in M_{n}(\mathbb{C})$, we have
    \begin{equation}
        \|\mathcal{E}X\mathcal{E}^{\dagger}\|_{\rm tr}\geq \|X\|_{\rm tr}+\sum\limits_{i=1}^{n}(|\gamma_{i}|^{2}-1)\braket{b_{i}|X|b_{i}}.
    \end{equation}
\end{lemma}

\begin{lemma}\label{lemma2}
    Suppose that $\ket{y_1},\cdots,\ket{y_n}$ are vectors in $\mathbb{C}^{d}$ and $G$ is their Gram matrix (i.e. $G$ is a matrix with entries $G_{ij}:=\braket{y_{i}|y_{j}}$). Then the operator $P:=\sum\limits_{i=1}^{n}\ket{y_{i}}\bra{y_{i}}$ is a projection if and only if $G^{2}=G$. 
\end{lemma}
\begin{proof}
    Suppose $G^{2}=G$. For any $\ket{v}\in \mathbb{C}^{d}$, we have $P^{2}\ket{v}=\sum\limits_{i=1}^{n}\sum\limits_{j=1}^{n}
    \braket{y_{i}|y_{j}}\braket{y_{j}|v}\ket{y_{i}}
    =\sum\limits_{i=1}^{n}\sum\limits_{j=1}^{n}G_{ij}\braket{y_{j}|v}\ket{y_{i}}$ and $P\ket{v}=\sum\limits_{j=1}^{n}\braket{y_{j}|v}\ket{y_{j}}$. Write $\ket{u}=P^{2}\ket{v}-P\ket{v}$, then $\ket{u}\in {\rm span}\{\ket{y_k}|k=1,2,\cdots,n\}$ and 
    \begin{flalign*}
    \braket{y_{k}|u}
    & =\sum\limits_{i=1}^{n}\sum\limits_{j=1}^{n}G_{ki}G_{ij}\braket{y_{j}|v}-\sum\limits_{j=1}^{n}G_{kj}\braket{y_{j}|v}=\sum\limits_{j=1}^{n}\left(\sum\limits_{i=1}^{n}G_{ki}G_{ij}\right)\braket{y_{j}|v}-\sum\limits_{j=1}^{n}G_{kj}\braket{y_{j}|v}\\
    & =\sum\limits_{j=1}^{n}(G^{2})_{kj}\braket{y_{j}|v}-\sum\limits_{j=1}^{n}G_{kj}\braket{y_{j}|v}=0
    \end{flalign*}
    for all $k\in\{1,2,\cdots,n\}$. And thus, $\ket{u}=0$, which means that $P^{2}\ket{v}=P\ket{v}$ holds for all $\ket{v}\in \mathbb{C}^{d}$. In other words, $P^{2}=P$. Since $P$ is also Hermitian, it is a projection.\par
    Conversely, suppose that $P$ is a projection. Then we have $P\ket{y_{k}}=\ket{y_{k}}$ ($k=1,\cdots,n$) because ${\rm ran}(P)={\rm span}\{\ket{y_{1}},\cdots,\ket{y_{n}}\}$ and $P^{2}=P$. Together with $P\ket{y_{k}}=\sum\limits_{i=1}^{n}G_{ik}\ket{y_{i}}$, we obtain $\ket{y_{k}}=\sum\limits_{i=1}^{n}G_{ik}\ket{y_{i}}$. Thus, $G_{jk}=\braket{y_{j}|y_{k}}=\sum\limits_{i=1}^{n}G_{ik}\braket{y_{j}|y_{i}}=\sum\limits_{i=1}^{n}G_{ji}G_{ik}=(G^{2})_{jk}$, which shows that $G^{2}=G$.
\end{proof}

Now, we can prove the central result of this section. 

\begin{theorem}\label{Result2}
    Let $\{E_{\alpha,k}|\alpha=1,\cdots,N;\,k=1,\cdots,M\}$ be an informationally complete $(N, M)$-POVM with the free parameter $x$ on the $d$ dimensional Hilbert space $\mathcal{H}$, $\rho$ be a bipartite state on $\mathcal{H}\otimes\mathcal{H}$, $\{\ket{w_{\alpha,k}}|\alpha=1,\cdots,N;\,k=1,\cdots,M\}$ be an orthonormal basis of $\mathbb{C}^{NM}$. Define $\mathcal{P}(\rho)=\sum\limits_{\alpha,\beta=1}^{N}\sum\limits_{k,l=1}^{M}{\rm tr}\left(\rho\left(E_{\alpha,k}\otimes E_{\beta,l}\right)\right)\ket{w_{\alpha,k}}\bra{w_{\beta,l}}$, and let $\mathcal{R}(\rho)$ be the realigned matrix \cite{CCNR1,CCNR2} of $\rho$, then 
    \begin{equation}\label{Eq:geqRealign}
        \dfrac{M(M-1)}{xM^{2}-d}\left(\|\mathcal{P}(\rho)\|_{\rm tr}-\dfrac{(d-1)(xM^{2}+d^{2})}{dM(M-1)}\right)\geq \|\mathcal{R}(\rho)\|_{\rm tr}-1\mbox{.}
    \end{equation}
\end{theorem}
\begin{proof}
    Denote $\ket{v_{\alpha,k}}=\sum\limits_{i=1}^{d}\sum\limits_{j=1}^{d}\braket{i|E_{\alpha,k}|j}\ket{ij}$, $\hat{\mathcal{E}}=\sum\limits_{\alpha=1}^{N}\sum\limits_{k=1}^{M}\ket{w_{\alpha,k}}\bra{v_{\alpha,k}}$, $X=\mathcal{R}(\rho)=\sum\limits_{p,q=1}^{d}\sum\limits_{a,b=1}^{d}\braket{pa|\rho|qb}\ket{pq}\bra{ab}$. Let $F$ be the swap operator on $\mathcal{H}\otimes\mathcal{H}$. Then 
    \begin{equation}
    \begin{split}
        \hat{\mathcal{E}}(XF)\hat{\mathcal{E}}^{\dagger} & =\sum\limits_{\alpha,\beta=1}^{N}\sum\limits_{k,l=1}^{M}\braket{v_{\alpha,k}|XF|v_{\beta,l}}\ket{w_{\alpha,k}}\bra{w_{\beta,l}}\\
        & =\sum\limits_{\alpha,\beta=1}^{N}\sum\limits_{k,l=1}^{M}\sum\limits_{i,j,s,t=1}^{d}\braket{j|E_{\alpha,k}|i}\braket{s|E_{\beta,l}|t}\braket{ij|X|ts}\ket{w_{\alpha,k}}\bra{w_{\beta,l}}\\
        & =\sum\limits_{\alpha,\beta=1}^{N}\sum\limits_{k,l=1}^{M}\sum\limits_{i,j,s,t=1}^{d}\braket{js|E_{\alpha,k}\otimes E_{\beta,l}|it}\braket{it|\rho|js}\ket{w_{\alpha,k}}\bra{w_{\beta,l}}\\
        & =\sum\limits_{\alpha,\beta=1}^{N}\sum\limits_{k,l=1}^{M}\sum\limits_{j,s=1}^{d}\braket{js|(E_{\alpha,k}\otimes E_{\beta,l})\rho|js}\ket{w_{\alpha,k}}\bra{w_{\beta,l}}\\
        & =\sum\limits_{\alpha,\beta=1}^{N}\sum\limits_{k,l=1}^{M}{\rm tr}\left(\left(E_{\alpha,k}\otimes E_{\beta,l}\right)\rho\right)\ket{w_{\alpha,k}}\bra{w_{\beta,l}}=\mathcal{P}(\rho). 
    \end{split}
    \end{equation}
    Therefore $\|\hat{\mathcal{E}}(XF)\hat{\mathcal{E}^{\dagger}}\|_{\rm tr}=\|\mathcal{P}(\rho)\|_{\rm tr}$. And $\|XF\|_{\rm tr}=\|X\|_{\rm tr}=\|\mathcal{R}(\rho)\|_{\rm tr}$, because $F$ is a unitary operator. \par
    Notice that $1={\rm tr}(\rho)=d\braket{\psi|X|\psi}$, where $\ket{\psi}=\frac{1}{\sqrt{d}}\sum\limits_{i=1}^{d}\ket{ii}$ is the maximally entangled state on $\mathcal{H}\otimes\mathcal{H}$, the inequality in the statement reads 
    \begin{equation}
        \dfrac{M(M-1)}{xM^{2}-d}\|\hat{\mathcal{E}}(XF)\hat{\mathcal{E}}^{\dagger}\|_{\rm tr}\geq \|X\|_{\rm tr}+\dfrac{d^{3}-xM^{2}}{xM^{2}-d}\braket{\psi|X|\psi}=\|XF\|_{\rm tr}+\dfrac{d^{3}-xM^{2}}{xM^{2}-d}\braket{\psi|XF|\psi}\mbox{.}
    \end{equation}
    In order to conclude, it remains to show that $\mathcal{E}:=\sqrt{\dfrac{M(M-1)}{xM^{2}-d}}\hat{\mathcal{E}}$ can be written as in equation (\ref{key}) from Lemma \ref{lemma1}, with $b_{1}=\psi,\gamma_{1}=\sqrt{\frac{d^{3}-xM^{2}}{xM^{2}-d}+1}=\sqrt{\frac{d(d^{2}-1)}{xM^{2}-d}},\gamma_{2}=\cdots=\gamma_{d^{2}}=1$.\par
     For all $\alpha\in\{1,2,\cdots,N\}$ and $k\in\{1,2,\cdots,M\}$, we define 
     \begin{equation}
        \ket{y_{\alpha,k}}:=\sqrt{\frac{M(M-1)}{xM^{2}-d}}\ket{v_{\alpha,k}}-\frac{1}{\sqrt{NM}}\left(\sqrt{\frac{d(d^{2}-1)}{xM^{2}-d}}-1\right)\ket{\psi}\mbox{.}
    \end{equation}
    And then, we have 
    \begin{equation}
        \mathcal{E}=\sqrt{\dfrac{M(M-1)}{xM^{2}-d}}\sum\limits_{\alpha=1}^{N}\sum\limits_{k=1}^{M}\ket{w_{\alpha,k}}\bra{v_{\alpha,k}}=\left(\sqrt{\frac{d(d^{2}-1)}{xM^{2}-d}}-1\right)\ket{v}\bra{\psi}+\sum\limits_{\alpha=1}^{N}\sum\limits_{k=1}^{M}\ket{w_{\alpha,k}}\bra{y_{\alpha,k}},
    \end{equation}
    where $\ket{v}=\dfrac{1}{\sqrt{NM}}\sum\limits_{\alpha=1}^{N}\sum\limits_{k=1}^{M}\ket{w_{\alpha,k}}$. Write $U:=\sum\limits_{\alpha=1}^{N}\sum\limits_{k=1}^{M}\ket{w_{\alpha,k}}\bra{y_{\alpha,k}}$, then $\mathcal{E}=\left(\sqrt{\frac{d(d^{2}-1)}{xM^{2}-d}}-1\right)\ket{v}\bra{\psi}+U$ and we can obtain $U\ket{\psi}=\ket{v}$ by direct computation. Clearly, $\braket{v_{\alpha,k}|v_{\beta,l}}={\rm tr}(E_{\alpha,k}E_{\beta,l})$ and $\braket{v_{\alpha,k}|\psi}=\frac{1}{\sqrt{d}}{\rm tr}(E_{\alpha,k})=\frac{\sqrt{d}}{M}$. Consequently, we can verify that $\braket{y_{\alpha,k}|y_{\beta,l}}=\frac{1}{NM}-\frac{1}{M}\delta_{\alpha\beta}+\delta_{\alpha\beta}\delta_{kl}$ by taking into account $N(M-1)=d^{2}-1$. Define $\|\ket{x}\|=\sqrt{\braket{x|x}}$. From $U^{\dagger}\ket{v}=\dfrac{1}{\sqrt{NM}}\sum\limits_{\alpha=1}^{N}\sum\limits_{k=1}^{M}\ket{y_{\alpha,k}}$, we obtain
    \begin{equation}
    \begin{split}
        \|\ket{\psi}-U^{\dagger}\ket{v}\|^{2} & = \braket{\psi|\psi}+\braket{v|UU^{\dagger}|v}-2{\rm Re}\braket{\psi|U^{\dagger}|v}\\
        & =\braket{\psi|\psi}+\frac{1}{NM}\sum\limits_{\alpha,\beta=1}^{N}\sum\limits_{k,l=1}^{M}\braket{y_{\alpha,k}|y_{\beta,l}}-2{\rm Re}\left(\frac{1}{\sqrt{NM}}\sum\limits_{\alpha=1}^{N}\sum\limits_{k=1}^{M}\braket{\psi|y_{\alpha,k}}\right)\\
        & =\braket{\psi|\psi}+\frac{1}{NM}\left(\frac{1}{NM}\times N^{2}M^{2}-\frac{1}{M}\times NM^{2}+NM\right)-2{\rm Re}\left(\frac{1}{\sqrt{NM}}\sum\limits_{\alpha=1}^{N}\sum\limits_{k=1}^{M}\frac{1}{\sqrt{NM}}\right)\\
        & =1+1-2=0, 
    \end{split}
    \end{equation}
    hence $U^{\dagger}\ket{v}=\ket{\psi}$. Let $G$ be the Gram matrix of $\{\ket{y_{\alpha,k}}|\alpha=1,\cdots,N;k=1,\cdots,M\}$. Then 
    \begin{flalign}
        (G^{2})_{\alpha k, \gamma m} & =\sum\limits_{\beta=1}^{N}\sum\limits_{l=1}^{M}G_{\alpha k,\beta l}G_{\beta l, \gamma m}=\sum\limits_{\beta=1}^{N}\sum\limits_{l=1}^{M}\left(\frac{1}{NM}+\delta_{\alpha\beta}\left(\delta_{kl}-\frac{1}{M}\right)\right)\left(\frac{1}{NM}+\delta_{\beta\gamma}\left(\delta_{lm}-\frac{1}{M}\right)\right)\\
        & =\frac{1}{NM}+\frac{1}{NM}\sum\limits_{\beta=1}^{N}\sum\limits_{l=1}^{M}\delta_{\beta\gamma}\left(\delta_{lm}-\frac{1}{M}\right)+\frac{1}{NM}\sum\limits_{\beta=1}^{N}\sum\limits_{l=1}^{M}\delta_{\alpha\beta}\left(\delta_{kl}-\frac{1}{M}\right)\nonumber\\
        & \qquad +\sum\limits_{\beta=1}^{N}\sum\limits_{l=1}^{M}\delta_{\alpha\beta}\delta_{\beta\gamma}\left(\delta_{kl}-\frac{1}{M}\right)\left(\delta_{lm}-\frac{1}{M}\right)\\
        & =\frac{1}{NM}+\frac{1}{NM}\sum\limits_{\beta=1}^{N}\delta_{\beta\gamma}\left(1-\frac{M}{M}\right)+\frac{1}{NM}\sum\limits_{\beta=1}^{N}\delta_{\alpha\beta}\left(1-\frac{M}{M}\right)+\delta_{\alpha\gamma}\sum\limits_{l=1}^{M}\left(\delta_{kl}-\frac{1}{M}\right)\left(\delta_{lm}-\frac{1}{M}\right)\\
        & =\frac{1}{NM}+\delta_{\alpha\gamma}\sum\limits_{l=1}^{M}\left(\delta_{kl}\delta_{lm}-\frac{1}{M}\delta_{kl}-\frac{1}{M}\delta_{lm}+\frac{1}{M^{2}}\right)\\
        & =\frac{1}{NM}+\delta_{\alpha\gamma}\left(\delta_{km}-\frac{1}{M}-\frac{1}{M}+\frac{1}{M}\right)=\frac{1}{NM}+\delta_{\alpha\gamma}\left(\delta_{km}-\frac{1}{M}\right)= G_{\alpha k,\gamma m}\mbox{,}
    \end{flalign}
    which shows that $G^{2}=G$. As a result, we have $U^{\dagger}U=\sum\limits_{\alpha=1}^{N}\sum\limits_{k=1}^{M}\ket{y_{\alpha,k}}\bra{y_{\alpha,k}}$ is a projection from Lemma \ref{lemma2}. \par
    Denote $V:=U-\ket{v}\bra{\psi}$, then we can see that $V\ket{\psi}=0$, $V^{\dagger}\ket{v}=0$ and $V^{\dagger}V=U^{\dagger}U-\ket{\psi}\bra{\psi}$. Thus, $V$ is a partial isometry which satisfies ${\rm ker}(V)^{\perp}\subseteq (\mathbb{C}\ket{\psi})^{\perp}$ and ${\rm ran}(V)={\rm ker}(V^{\dagger})^{\perp}\subseteq(\mathbb{C}\ket{v})^{\perp}$, where ${\rm ran}(V)$ denotes the range of $V$. Combining this with $\mathcal{E}=\sqrt{\frac{d(d^{2}-1)}{xM^{2}-d}}\ket{v}\bra{\psi}+V$ and applying Lemma \ref{lemma_pariso} completes the proof.
\end{proof}


Theorem \ref{Result2} shows that the lower bound of concurrence based on symmetric measurements \cite{HFWang1} is always tighter than the one based on realignment \cite{LowBound0}. By applying appropriate monotonically increasing functions to both sides of Eq.\eqref{Eq:geqRealign}, we also obtain that the lower bounds of $q$-concurrence and $\alpha$-concurrence based on symmetric measurements \cite{NYang} are always tighter than the one based on realignment \cite{LBqc1, LBqc2, LBalphac}. In addition, it likewise shows that the symmetric measurement-based criterion for Schmidt numbers \cite{HFWang2} is always stronger than the realignment criterion for Schmidt numbers \cite{CCNRschmidt}. \par

Theorem \ref{Result2} implies a more profound problem. Since realignment can also be understood from the perspective of conical 2-designs, Theorem \ref{Result2} actually compares the lower bounds of concurrence induced by two specific types of conical 2-designs. Consequently, we naturally have the following conjecture.
\begin{conjecture}
    The lower bounds of concurrence induced by any two distinct conical 2-designs are comparable. In other words, define $B_{E}(\rho)=\beta_{E}^{-1}\sqrt{\frac{2}{d(d-1)}}\left(\|\mathcal{P}_{E}(\rho)\|_{\rm tr}-\alpha_{E}-\beta_{E}\right)$ for any conical 2-design $E$, then 
    \begin{equation*}
        \{B_{E}| E\enspace\text{is a conical 2-design}\}
    \end{equation*}
    is a totally ordered set with respect to the relation $\preceq$, defined by $B_{G}\preceq B_{E}$ if and only if $B_{G}(\rho)\leq B_{E}(\rho)$ for every bipartite state $\rho$ on $\mathcal{H}\otimes\mathcal{H}$.
\end{conjecture}
If the above conjecture is true, then a natural question is whether $\{B_{E}| E\enspace\text{is a conical 2-design}\}$ has a maximum element. Equivalently, do there exist conical 2-designs achieving the tightest lower bound of concurrence. Moreover, if this tightest lower bound exists, can it detect all PPT entangled states?

\section{Improved lower bounds of concurrence induced by symmetric measurements} \label{sec4}
In this section, we derive an extended version of Lemma \ref{NMLOWBOUND_Lemma}, which gives an improved lower bound of concurrence based on symmetric measurements. Let $\mathcal{H}_{A}$ and $\mathcal{H}_{B}$ be $d_{A}$-dimensional and $d_{B}$-dimensional Hilbert spaces, respectively. We consider a bipartite state $\rho$ acting on $\mathcal{H}_{A}\otimes\mathcal{H}_{B}$, and two sets of informationally complete symmetric measurements, which are $\{E_{\alpha,k}^{A}|\alpha=1,\cdots,N_{A};\,k=1,\cdots,M_{A}\}$ with the free parameter $x_{A}$ on $\mathcal{H}_{A}$ and $\{E_{\beta,l}^{B}|\beta=1,\cdots,N_{B};\,l=1,\cdots,M_{B}\}$ with the free parameter $x_{B}$ on $\mathcal{H}_{B}$. With respect to the measurement operators $E_{\alpha,k}^{A}\otimes E_{\beta,l}^{B}$, we denote the probabilities $p_{\alpha,k;\beta,l}={\rm tr}\left(\rho\left(E_{\alpha,k}^{A}\otimes E_{\beta,l}^{B}\right)\right)$ and define the matrix $\mathcal{P}(\rho)=\left(p_{\alpha,k;\beta,l}\right)_{N_{A}M_{A}\times N_{B}M_{B}}$. We further define \cite{YLu}
\begin{equation}\label{key matrix}
    \mathcal{M}_{\mu, \nu}(\rho)=\begin{pmatrix}
	\mu\nu^{\mathsf{T}} & \mu b_{\rho}^{\mathsf{T}}\\
	a_{\rho}\nu^{\mathsf{T}} & \mathcal{P}(\rho)
\end{pmatrix}, 
\end{equation}
where $\mu\in\mathbb{R}^{m}$, $\nu\in\mathbb{R}^{n}$, $a_{\rho}=\begin{pmatrix}
{\rm tr}(\rho_{A}E_{1,1}^{A})\\
{\rm tr}(\rho_{A}E_{1,2}^{A})\\
\cdots\\
{\rm tr}(\rho_{A}E_{N_{A},M_{A}}^{A})
\end{pmatrix}$, $b_{\rho}=\begin{pmatrix}
{\rm tr}(\rho_{B}E_{1,1}^{B})\\
{\rm tr}(\rho_{B}E_{1,2}^{B})\\
\cdots\\
{\rm tr}(\rho_{B}E_{N_{B},M_{B}}^{B})
\end{pmatrix}$, $\rho_{A}={\rm tr}_{B}(\rho)$, $\rho_{B}={\rm tr}_{A}(\rho)$.

\begin{theorem}\label{Result1}
    For any bipartite state $\rho$ acting on $\mathcal{H}_{A}\otimes\mathcal{H}_{B}$, the concurrence $C(\rho)$ satisfies
    \begin{equation}\label{result1}
    C(\rho)\geq K\sqrt{\dfrac{2}{d(d-1)}}\left(\|\mathcal{M}_{\mu,\nu}(\rho)\|_{\rm tr}-\sqrt{|\mu|^{2}+\dfrac{(d_{A}-1)(M_{A}^{2}x_{A}+d_{A}^{2})}{d_{A}M_{A}(M_{A}-1)}}\sqrt{|\nu|^{2}+\dfrac{(d_{B}-1)(M_{B}^{2}x_{B}+d_{B}^{2})}{d_{B}M_{B}(M_{B}-1)}}\right),
    \end{equation}
    where $K=\sqrt{\frac{M_{A}(M_{A}-1)M_{B}(M_{B}-1)}{(x_{A}M_{A}^{2}-d_{A})(x_{B}M_{B}^{2}-d_{B})}}$ and $d=\min\{d_{A},d_{B}\}$.
\end{theorem}
\begin{proof}
    Let $\{p_i,\ket{\psi_i}\}$ be an optimal pure state decomposition of $\rho$ such that $C(\rho)=\sum\limits_{i}p_{i}C(\ket{\psi_i})$. Noting that $\sum\limits_{i}p_{i}\|\mathcal{M}_{\mu,\nu}(\ket{\psi_{i}}\bra{\psi_{i}})\|_{\rm tr}\geq \|\mathcal{M}_{\mu,\nu}(\rho)\|_{\rm tr}$ by the convex property of the trace norm, we only need to prove the theorem for the case of pure states $\ket{\psi}\in \mathcal{H}_{A}\otimes \mathcal{H}_{B}$,
    \begin{equation*}
     C(\ket{\psi})\geq K\sqrt{\dfrac{2}{d(d-1)}}
     \left(\|\mathcal{M_{\mu,\nu}}(\ket{\psi}\bra{\psi})\|_{\rm tr}-\sqrt{|\mu|^{2}+\dfrac{(d_{A}-1)(M_{A}^{2}x_{A}+d_{A}^{2})}{d_{A}M_{A}(M_{A}-1)}}\sqrt{|\nu|^{2}+\dfrac{(d_{B}-1)(M_{B}^{2}x_{B}+d_{B}^{2})}{d_{B}M_{B}(M_{B}-1)}}\right).
    \end{equation*}
    \noindent By the Schmidt decomposition, there exists a positive integer $r\leq d$,  and orthonormal bases $\{\ket{e_{i}}\}_{i=1}^{d_{A}}$ in $\mathcal{H}_{A}$ and $\{\ket{f_{j}}\}_{j=1}^{d_{B}}$ in $\mathcal{H}_{B}$ such that $\ket{\psi}=\sum\limits_{i=1}^{r}\lambda_{i}\ket{e_{i}}\otimes\ket{f_{i}}$, where $\lambda_{i}>0$ and $\sum\limits_{i=1}^{r}\lambda_{i}^{2}=1$. Denote $\sigma=\sum\limits_{i=1}^{r}\lambda_{i}^{2}\ket{e_{i}f_{i}}\bra{e_{i}f_{i}}$. Obviously, we have ${\rm tr}_{B}(\sigma)={\rm tr}_{B}(\ket{\psi}\bra{\psi})$ and ${\rm tr}_{A}(\sigma)={\rm tr}_{A}(\ket{\psi}\bra{\psi})$. Then $\mathcal{M}_{\mu, \nu}(\ket{\psi}\bra{\psi})=\mathcal{M}_{\mu, \nu}(\sigma)+\begin{pmatrix}
	0 & 0\\
	0 & \mathcal{P}(\ket{\psi}\bra{\psi})-\mathcal{P}(\sigma)
    \end{pmatrix}$, and consequently
    \begin{equation}\label{ineq_keydecom}
        \|\mathcal{M}_{\mu, \nu}(\ket{\psi}\bra{\psi})\|_{\rm tr}\leq \|\mathcal{M}_{\mu, \nu}(\sigma)\|_{\rm tr}+\|\mathcal{P}(\ket{\psi}\bra{\psi})-\mathcal{P}(\sigma)\|_{\rm tr}.
    \end{equation}
    Let $\{\ket{h_{\alpha,k}}|\alpha=1,\cdots,N_{A};k=1,\cdots,M_{A}\}$ and $\{\ket{w_{\beta,l}}|\beta=1,\cdots,N_{B};l=1,\cdots,M_{B}\}$ be orthonormal bases of $\mathbb{C}^{N_{A}M_{A}}$ and $\mathbb{C}^{N_{B}M_{B}}$, respectively. Define $\ket{u_{ij}}=\sum\limits_{\alpha=1}^{N_{A}}\sum\limits_{k=1}^{M_{A}}\braket{e_j|E_{\alpha,k}^{A}|e_{i}}\ket{h_{\alpha,k}}$, $\ket{v_{ij}}=\sum\limits_{\beta=1}^{N_{B}}\sum\limits_{l=1}^{M_{B}}\braket{f_j|E_{\beta,l}^{B}|f_{i}}\ket{w_{\beta,l}}$, $\mathcal{P}^{\prime}(\rho)=\sum\limits_{\alpha=1}^{N_{A}}\sum\limits_{\beta=1}^{N_{B}}\sum\limits_{k=1}^{M_{A}}\sum\limits_{l=1}^{M_{B}}{\rm tr}(\rho(E_{\alpha,k}^{A}\otimes E_{\beta,l}^{B}))\ket{h_{\alpha,k}}\bra{w_{\beta,l}}$. By direct calculation, we have
    \begin{equation}
    \begin{split}
        \mathcal{P}^{\prime}(\ket{\psi}\bra{\psi}) & =\sum\limits_{i=1}^{r}\sum\limits_{j=1}^{r}\lambda_{i}\lambda_{j}\ket{u_{ij}}\bra{\bar{v}_{ij}},\\
        \mathcal{P}^{\prime}(\sigma)& =\sum\limits_{i=1}^{r}\lambda_{i}^{2}\ket{u_{ii}}\bra{\bar{v}_{ii}},\\
         \mathcal{P}^{\prime}(\ket{\psi}\bra{\psi})-\mathcal{P}^{\prime}(\sigma) & =\sum\limits_{i\neq j}\lambda_{i}\lambda_{j}\ket{u_{ij}}\bra{\bar{v}_{ij}},
    \end{split}
    \end{equation}
    where $\ket{\bar{v}_{ij}}=\sum\limits_{\beta=1}^{N_{B}}\sum\limits_{l=1}^{M_{B}}\overline{\braket{f_{j}|E_{\beta,l}^{B}|f_i}}\ket{w_{\beta,l}}$ with $\overline{\braket{f_{j}|E_{\beta,l}^{B}|f_i}}$ standing for complex conjugation of $\braket{f_{j}|E_{\beta,l}^{B}|f_i}$.\par
    From Lemma \ref{NMconi2_Lemma}, we obtain $\sum\limits_{\alpha=1}^{N_{A}}\sum\limits_{k=1}^{M_{A}}E_{\alpha,k}^{A}\otimes E_{\alpha,k}^{A}=\frac{x_{A}M_{A}^{2}-d_{A}}{M_{A}(M_{A}-1)}F+\frac{d_{A}^{3}-x_{A}M_{A}^{2}}{d_{A}M_{A}(M_{A}-1)}I$. Hence, we have 
    \begin{equation}\label{011}
    \braket{u_{ij}|u_{i^{\prime}j^{\prime}}}=\braket{e_{i}e_{j^{\prime}}|\sum\limits_{\alpha=1}^{N_{A}}\sum\limits_{k=1}^{M_{A}}E_{\alpha,k}^{A}\otimes E_{\alpha,k}^{A}|e_{j}e_{i^{\prime}}}=\dfrac{x_{A}M_{A}^{2}-d_{A}}{M_{A}(M_{A}-1)}\delta_{ii^{\prime}}\delta_{jj^{\prime}}
    +\dfrac{d_{A}^{3}-x_{A}M_{A}^{2}}{d_{A}M_{A}(M_{A}-1)}\delta_{ij}\delta_{i^{\prime}j^{\prime}}.
    \end{equation}
    Similarly, due to $\sum\limits_{\beta=1}^{N_{B}}\sum\limits_{l=1}^{M_{B}}E_{\beta,l}^{B}\otimes E_{\beta,l}^{B}=\frac{x_{B}M_{B}^{2}-d_{B}}{M_{B}(M_{B}-1)}F+\frac{d_{B}^{3}-x_{B}M_{B}^{2}}{d_{B}M_{B}(M_{B}-1)}I$, we have $\braket{v_{ij}|v_{i^{\prime}j^{\prime}}}=\frac{x_{B}M_{B}^{2}-d_{B}}{M_{B}(M_{B}-1)}
    \delta_{ii^{\prime}}\delta_{jj^{\prime}}+\frac{d_{B}^{3}-x_{B}M_{B}^{2}}{d_{B}M_{B}(M_{B}-1)}
    \delta_{ij}\delta_{i^{\prime}j^{\prime}}\in\mathbb{R}$. With help of this and $\braket{\overline{v}_{ij}|\overline{v}_{i^{\prime}j^{\prime}}}
    =\overline{\braket{v_{ij}|v_{i^{\prime}j^{\prime}}}}$, we obtain 
    \begin{equation}\label{012}
    \braket{\overline{v}_{ij}|\overline{v}_{i^{\prime}j^{\prime}}}=\dfrac{x_{B}M_{B}^{2}-d_{B}}{M_{B}(M_{B}-1)}\delta_{ii^{\prime}}\delta_{jj^{\prime}}
    +\dfrac{d_{B}^{3}-x_{B}M_{B}^{2}}{d_{B}M_{B}(M_{B}-1)}\delta_{ij}\delta_{i^{\prime}j^{\prime}}
    \end{equation}
    Therefore, both $\{\ket{u_{ij}}|i\neq j\}$ and $\{\ket{\overline{v}_{ij}}|i\neq j\}$ are orthogonal sets. So we have
    \begin{equation*}
        [\mathcal{P}^{\prime}(\ket{\psi}\bra{\psi})-\mathcal{P}^{\prime}(\sigma)]^{\dagger}[\mathcal{P}^{\prime}(\ket{\psi}\bra{\psi})-\mathcal{P}^{\prime}(\sigma)]=\sum\limits_{i\neq j}(\lambda_{i}\lambda_{j})^{2}\braket{u_{ij}|u_{ij}}\ket{\bar{v}_{ij}}\bra{\bar{v}_{ij}}.
    \end{equation*}
    Hence, 
    \begin{equation}
    \begin{split}
        \|\mathcal{P}(\ket{\psi}\bra{\psi})-\mathcal{P}(\sigma)\|_{\rm tr} & =\|\mathcal{P}^{\prime}(\ket{\psi}\bra{\psi})-\mathcal{P}^{\prime}(\sigma)\|_{\rm tr}=\sum\limits_{i\neq j}\lambda_{i}\lambda_{j}\sqrt{\braket{u_{ij}|u_{ij}}}\sqrt{\braket{\bar{v}_{ij}|\bar{v}_{ij}}}\\
        & =\frac{2}{K}\sum\limits_{i<j}\lambda_{i}\lambda_{j},
    \end{split}
    \end{equation}
    where the last equality is from \eqref{011} and \eqref{012}. Combining this with \eqref{ineq_keydecom}, we have
    \begin{equation}
        \|\mathcal{M}_{\mu, \nu}(\ket{\psi}\bra{\psi})\|_{\rm tr}\leq \|\mathcal{M}_{\mu, \nu}(\sigma)\|_{\rm tr}+\frac{2}{K}\sum\limits_{i<j}\lambda_{i}\lambda_{j}.
    \end{equation}
    According to Theorem 1 in Ref.\cite{YLu}, we obtain
    \begin{equation}\label{021}
        \|\mathcal{M}_{\mu, \nu}(\ket{\psi}\bra{\psi})\|_{\rm tr}\leq  \sqrt{|\mu|^{2}+\dfrac{(d_{A}-1)(M_{A}^{2}x_{A}+d_{A}^{2})}{d_{A}M_{A}(M_{A}-1)}}\sqrt{|\nu|^{2}+\dfrac{(d_{B}-1)(M_{B}^{2}x_{B}+d_{B}^{2})}{d_{B}M_{B}(M_{B}-1)}}
        +\frac{2}{K}\sum\limits_{i<j}\lambda_{i}\lambda_{j}.
    \end{equation}
    With the help of (\ref{021}) and $C(\ket{\psi})\geq 2\sqrt{\dfrac{2}{d(d-1)}}\sum\limits_{i<j}\lambda_{i}\lambda_{j}$\cite{LowBound0}, we can finish the proof. 
\end{proof}

Taking $\mu=0$ and $\nu=0$ in Theorem \ref{Result1}, we have
\begin{equation*}
	C(\rho)\geq K\sqrt{\dfrac{2}{d(d-1)}}\left(\|\mathcal{P}(\rho)\|_{\rm tr}-
    \sqrt{\dfrac{(d_{A}-1)(x_{A}M_{A}^{2}+d_{A}^{2})}{d_{A}M_{A}(M_{A}-1)}}\sqrt{\dfrac{(d_{B}-1)(x_{B}M_{B}^{2}+d_{B}^{2})}{d_{B}M_{B}(M_{B}-1)}}\right).
\end{equation*}
Hence, Theorem \ref{Result1} is a generalization of Lemma \ref{NMLOWBOUND_Lemma}. To illustrate the advantages of Theorem \ref{Result1}, we give the following example.

\begin{example}
    Consider the mixture of the bound entangled state proposed by Horodecki\cite{example1},
	\begin{equation*}
		\rho_{\tau}=\dfrac{1}{1+8\tau}
		\begin{pmatrix}
			\tau & 0 & 0 & 0 & \tau & 0 & 0 & 0 & \tau\\
			0 & \tau & 0 & 0 & 0 & 0 & 0 & 0 & 0\\
			0 & 0 & \tau & 0 & 0 & 0 & 0 & 0 & 0\\
			0 & 0 & 0 & \tau & 0 & 0 & 0 & 0 & 0\\
			\tau & 0 & 0 & 0 & \tau & 0 & 0 & 0 & \tau\\
			0 & 0 & 0 & 0 & 0 & \tau & 0 & 0 & 0\\
			0 & 0 & 0 & 0 & 0 & 0 & \frac{1+\tau}{2} & 0 & \frac{\sqrt{1-\tau^{2}}}{2}\\
			0 & 0 & 0 & 0 & 0 & 0 & 0 & 0 & 0\\
			\tau & 0 & 0 & 0 & \tau & 0 & \frac{\sqrt{1-\tau^{2}}}{2} & 0 & \frac{1+\tau}{2}
		\end{pmatrix},\enspace 0<\tau<1,
	\end{equation*}
    and the $9\times 9$ identity matrix $I_{9}$,
    \begin{equation*}
    	\rho(\tau,q)=q\rho_{\tau}+\frac{1-q}{9}I_{9}.
    \end{equation*}
    We construct a $(8,2)$-POVM with the Hermitian basis operators $G_{\alpha,k}$ given in Appendix \ref{Appendix A}. It is verified that the parameter $x=\frac{3}{4}+t^{2}(\sqrt{2}+1)^{2}$ with $t\in[-0.2536,0.2536]$. Fig.\ref{fig1} illustrates the lower bounds of $C(\rho(\tau,0.995))$. The purple curve is the bound from Theorem \ref{Result1} for $\mu=(0.47,-0.91)^{\mathsf{T}},\nu=(-0.97,-0.02)^{\mathsf{T}}$, and the $(8,2)$-POVM with $t=0.01$, which shows that $\rho(\tau,0.995)$ is entangled for $0.102705\leq\tau\leq 0.458084$. The red curve is the bound from Theorem in Ref.\cite{HFWang1} based on the $(8,2)$-POVM with $t=0.01$, which shows that $\rho(\tau,0.995)$ is entangled for $0.111177\leq\tau\leq 0.453692$. Clearly, the purple curve is a tighter lower bound and detects more entangled states.
    \begin{figure}[h]
    	\includegraphics[width=10cm]{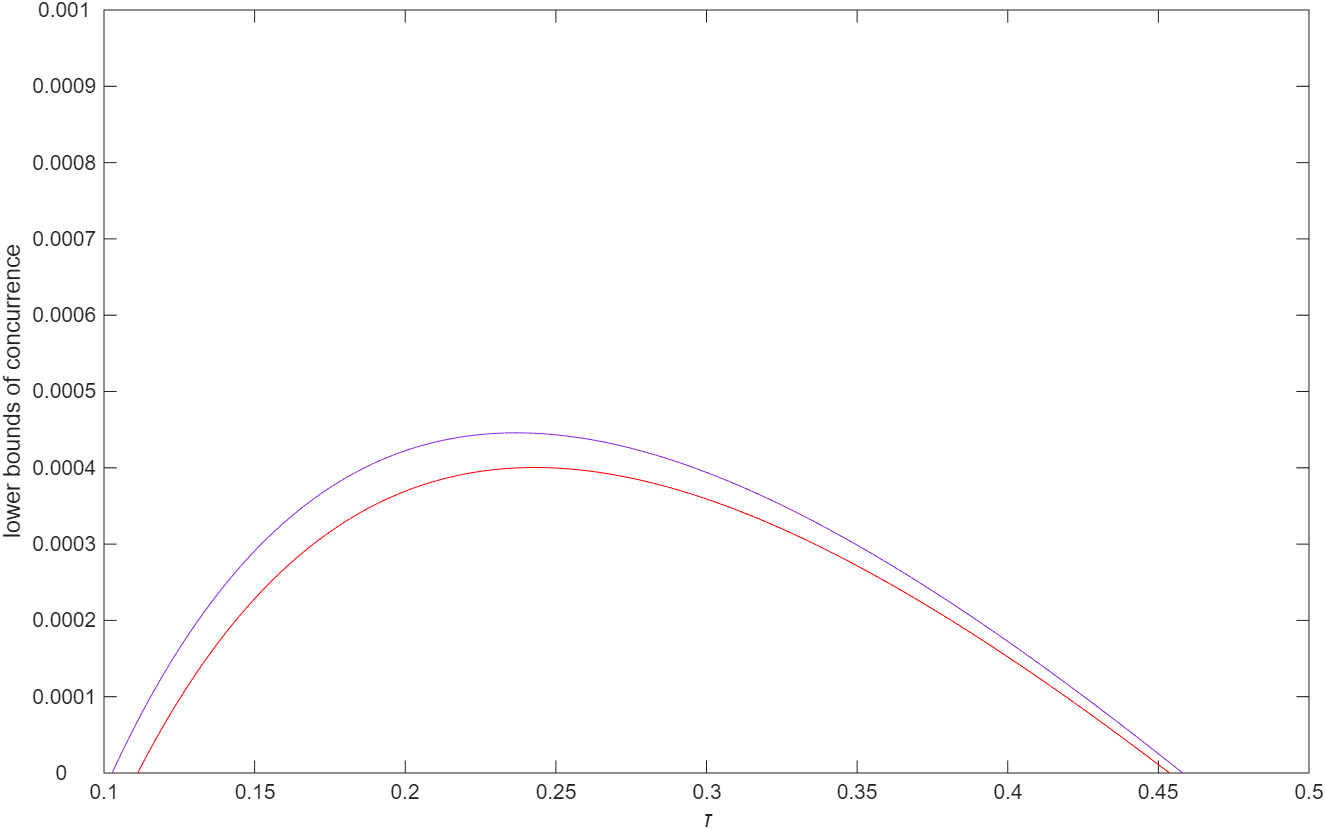}
    	\caption{Lower bounds of $C(\rho(\tau,0.995))$. The purple curve for the bound from Theorem \ref{Result1}. The red curve for the bound from Theorem in Ref.\cite{HFWang1}.}\label{fig1}
    \end{figure}
\end{example}

Based on Theorem \ref{Result1}, we can present a lower bound of genuine tripartite entanglement concurrence. The genuine multipartite entanglement (GME) concurrence of a tripartite pure state $\ket{\varphi}\in\mathbb{C}^{d}\otimes\mathbb{C}^{d}\otimes\mathbb{C}^{d}$ is defined by \cite{GMEC}, 
\begin{equation*}
    C_{\rm GME}(\ket{\varphi})=\sqrt{\min\{1-{\rm tr}(\rho_{1}^{2}),1-{\rm tr}(\rho_{2}^{2}),1-{\rm tr}(\rho_{3}^{2})\}},
\end{equation*}
where $\rho_{i}$ is the reduced density matrix of the subsystem $i$. The GME concurrence of a mixed state $\rho$ is defined as
\begin{equation*}
    C_{\rm GME}(\rho)=\min_{\{p_{i},\ket{\varphi_{i}}\}}\sum\limits_{i}p_{i}C_{\rm GME}(\ket{\varphi_{i}}), 
\end{equation*}
where the minimum is taken over all possible ensemble decompositions of $\rho=\sum\limits_{i}p_{i}\ket{\varphi_{i}}\bra{\varphi_{i}}$, $p_{i}\geq 0$ with $\sum\limits_{i}p_{i}=1$. Combining Theorem 1 of Ref.\cite{GMECLB} with our Theorem \ref{Result1} immediately yields the following result.

\begin{corollary}
    Let $\{E_{\alpha,k}^{A}|\alpha=1,\cdots,N_{A};\,k=1,\cdots,M_{A}\}$ be an informationally complete $(N_{A}, M_{A})$-POVM with the free parameter $x_{A}$ on $\mathbb{C}^{d}$, $\{E_{\beta,l}^{B}|\beta=1,\cdots,N_{B};\,l=1,\cdots,M_{B}\}$ be an informationally complete $(N_{B}, M_{B})$-POVM with the free parameter $x_{B}$ on $\mathbb{C}^{d}\otimes\mathbb{C}^{d}$, $\rho$ be a tripartite state on $\mathbb{C}^{d}\otimes\mathbb{C}^{d}\otimes\mathbb{C}^{d}$. Denote 
    \begin{equation*}
        f_{\mu,\nu}(\rho)=\frac{1}{3}\left(\|\mathcal{M}_{\mu,\nu}^{1|23}(\rho)\|_{\rm tr}+\|\mathcal{M}_{\mu,\nu}^{2|13}(\rho)\|_{\rm tr}+\|\mathcal{M}_{\mu,\nu}^{3|12}(\rho)\|_{\rm tr}\right),
    \end{equation*}
    where $\mathcal{M}_{\mu,\nu}^{i|jk}$ stands for the matrix (\ref{key matrix}) under bipartition $i$ and $jk$, $\{i,j,k\}=\{1,2,3\}$. Then we have 
    \begin{equation*}
        C_{\rm GME}(\rho)\geq \frac{3}{\sqrt{d(d-1)}}K\left(f_{\mu,\nu}(\rho)-L\right)-2\sqrt{\frac{d-1}{d}},
    \end{equation*}
    where $K=\sqrt{\frac{M_{A}(M_{A}-1)M_{B}(M_{B}-1)}{(x_{A}M_{A}^{2}-d)(x_{B}M_{B}^{2}-d^{2})}}$ and $L=\sqrt{|\mu|^{2}+\frac{(d-1)(M_{A}^{2}x_{A}+d^{2})}{dM_{A}(M_{A}-1)}}\sqrt{|\nu|^{2}+\frac{(d^{2}-1)(M_{B}^{2}x_{B}+d^{4})}{d^{2}M_{B}(M_{B}-1)}}$.
\end{corollary}

\section{CONCLUSIONS AND DISCUSSIONS} \label{sec5}

We have proved that, for arbitrary bipartite states, the lower bound of concurrence induced by symmetric measurements is always tighter than the one obtained from realignment. This explains why several previous results based on symmetric measurements \cite{HFWang2, NYang} have been found to be stronger than those based on realignment \cite{CCNRschmidt, LBqc1, LBqc2, LBalphac}. Furthermore, we have conjectured that the essential reason behind this may be that the lower bounds of concurrence induced by arbitrary two distinct conical 2-designs are comparable. In addition, we have derived an extended version of lower bounds of concurrence induced by symmetric measurements, which can be tighter than the original one for some states and can also be experimentally identified without state tomography. As a direct corollary, we have also presented a corresponding lower bound of genuine tripartite entanglement concurrence.\par

If the lower bounds of concurrence induced by arbitrary two distinct conical 2-designs are comparable, then a natural problem is whether there exist conical 2-designs achieving the tightest lower bound of concurrence. We believe this deserves further study. Another interesting problem is for which states the lower bounds of concurrence induced by conical 2‑designs attain the exact value of concurrence. As shown in Ref.\cite{HFWang1}, the lower bounds of concurrence induced by symmetric measurements reach saturation for isotropic states. This is unlikely to be a coincidence: we guess that there exist some states such that the lower bound of concurrence induced by arbitrary conical 2-design reaches saturation. \par

\bigskip 
    \noindent{\bf Acknowledgements}
This work is supported by the National Natural Science Foundation of China (NSFC) under Grants No.12471427 and No.12171044. Kimi K3 was used solely to polish language.

\begin{appendices}
\section{A proof of the inequality $C(\rho)\geq B_E(\rho)$}\label{Appendix B_{E}}
{\bf The proof of $C(\rho)\geq B_E(\rho)$}: Let $\{p_i,\ket{\psi_i}\}$ be an optimal pure state decomposition of $\rho$ such that $C(\rho)=\sum\limits_{i}p_{i}C(\ket{\psi_i})$. Noting that $\sum\limits_{i}p_{i}\|\mathcal{P}(\ket{\psi_{i}}\bra{\psi_{i}})\|_{\rm tr}\geq \|\mathcal{P}(\rho)\|_{\rm tr}$ by the convex property of the trace norm, we only need to prove the inequality for the case of pure states $\ket{\psi}\in \mathcal{H}\otimes \mathcal{H}$,
    \begin{equation*}
     C(\ket{\psi})\geq \beta_{E}^{-1}\sqrt{\frac{2}{d(d-1)}}\left(\|\mathcal{P}_{E}(\ket{\psi}\bra{\psi})\|_{\rm tr}-\beta_{E}-\alpha_{E}\right).
    \end{equation*}
By the Schmidt decomposition, there exists a positive integer $r\leq d$, and orthonormal bases $\{\ket{e_{i}}\}_{i=1}^{d}$ and $\{\ket{f_{j}}\}_{j=1}^{d}$ in $\mathcal{H}_{A}$ such that $\ket{\psi}=\sum\limits_{i=1}^{r}\lambda_{i}\ket{e_{i}}\otimes\ket{f_{i}}$, where $\lambda_{i}>0$ and $\sum\limits_{i=1}^{r}\lambda_{i}^{2}=1$. Define $\ket{u_{ij}}=\sum\limits_{k=1}^{m}\braket{e_j|E_{k}|e_{i}}\ket{w_{k}}$, $\ket{v_{ij}}=\sum\limits_{l=1}^{m}\braket{f_j|E_{l}|f_{i}}\ket{w_{l}}$, and $\sigma=\sum\limits_{i=1}^{r}\lambda_{i}^{2}\ket{e_{i}f_{i}}\bra{e_{i}f_{i}}$. By direct calculation, we have
    \begin{equation*}
    \begin{split}
        \mathcal{P}_{E}(\ket{\psi}\bra{\psi}) & =\sum\limits_{i=1}^{r}\sum\limits_{j=1}^{r}\lambda_{i}\lambda_{j}\ket{u_{ij}}\bra{\bar{v}_{ij}},\\
        \mathcal{P}_{E}(\sigma)& =\sum\limits_{i=1}^{r}\lambda_{i}^{2}\ket{u_{ii}}\bra{\bar{v}_{ii}},\\
         \mathcal{P}_{E}(\ket{\psi}\bra{\psi})-\mathcal{P}_{E}(\sigma) & =\sum\limits_{i\neq j}\lambda_{i}\lambda_{j}\ket{u_{ij}}\bra{\bar{v}_{ij}},
    \end{split}
    \end{equation*}
        where $\ket{\bar{v}_{ij}}=\sum\limits_{l=1}^{m}\overline{\braket{f_{j}|E_{l}|f_i}}\ket{w_{l}}$ with $\overline{\braket{f_{j}|E_{l}|f_i}}$ standing for complex conjugation of $\braket{f_{j}|E_{l}|f_i}$.\par
    Since $\sum\limits_{k=1}^{m}E_{k}\otimes E_{k}=\alpha_{E}I+\beta_{E}F$, we have 
    \begin{equation}\label{appendix 011}
    \braket{u_{ij}|u_{i^{\prime}j^{\prime}}}
    =\braket{e_{i}e_{j^{\prime}}|\sum\limits_{k=1}^{m}E_{k}\otimes E_{k}|e_{j}e_{i^{\prime}}}=\beta_{E}\delta_{ii^{\prime}}\delta_{jj^{\prime}}
    +\alpha_{E}\delta_{ij}\delta_{i^{\prime}j^{\prime}}.
    \end{equation}
    Similarly, we have $\braket{v_{ij}|v_{i^{\prime}j^{\prime}}}=\beta_{E}\delta_{ii^{\prime}}\delta_{jj^{\prime}}+\alpha_{E}\delta_{ij}\delta_{i^{\prime}j^{\prime}}\in\mathbb{R}$. With help of this and $\braket{\overline{v}_{ij}|\overline{v}_{i^{\prime}j^{\prime}}}
    =\overline{\braket{v_{ij}|v_{i^{\prime}j^{\prime}}}}$, we obtain 
    \begin{equation}\label{appendix 012}
    \braket{\overline{v}_{ij}|\overline{v}_{i^{\prime}j^{\prime}}}
    =\beta_{E}\delta_{ii^{\prime}}\delta_{jj^{\prime}}
    +\alpha_{E}\delta_{ij}\delta_{i^{\prime}j^{\prime}}.
    \end{equation}
    Therefore, both $\{\ket{u_{ij}}|i\neq j\}$ and $\{\ket{\overline{v}_{ij}}|i\neq j\}$ are orthogonal sets. So we have
    \begin{equation*}
        [\mathcal{P}_{E}(\ket{\psi}\bra{\psi})-\mathcal{P}_{E}(\sigma)]^{\dagger}[\mathcal{P}_{E}(\ket{\psi}\bra{\psi})-\mathcal{P}_{E}(\sigma)]=\sum\limits_{i\neq j}(\lambda_{i}\lambda_{j})^{2}\braket{u_{ij}|u_{ij}}\ket{\bar{v}_{ij}}\bra{\bar{v}_{ij}}.
    \end{equation*}
    Hence, 
    \begin{equation}
    \|\mathcal{P}_{E}(\ket{\psi}\bra{\psi})-\mathcal{P}_{E}(\sigma)\|_{\rm tr}=\sum\limits_{i\neq j}\lambda_{i}\lambda_{j}\sqrt{\braket{u_{ij}|u_{ij}}}\sqrt{\braket{\bar{v}_{ij}|\bar{v}_{ij}}}=2\beta_{E}\sum\limits_{i<j}\lambda_{i}\lambda_{j},
    \end{equation}
    where the last equality is from (\ref{appendix 011}) and (\ref{appendix 012}). On the other hand, we can see that
    \begin{equation*}
        \mathcal{P}_{E}(\ket{e_{i}f_{i}}\bra{e_{i}f_{i}})^{\dagger}\mathcal{P}_{E}(\ket{e_{i}f_{i}}\bra{e_{i}f_{i}})=(\ket{u_{ii}}\bra{\bar{v}_{ii}})^{\dagger}(\ket{u_{ii}}\bra{\bar{v}_{ii}})=\braket{u_{ii}|u_{ii}}\ket{\bar{v}_{ii}}\bra{\bar{v}_{ii}}
    \end{equation*}
    Hence,
    \begin{equation}
        \|\mathcal{P}_{E}(\sigma)\|_{\rm tr}\leq\sum_{i=1}^{r}\lambda_{i}^{2}\|\mathcal{P}_{E}(\ket{e_{i}f_{i}}\bra{e_{i}f_{i}})\|_{\rm tr}=\sum\limits_{i=1}^{r}\lambda_{i}^{2}\sqrt{\braket{u_{ii}|u_{ii}}}\sqrt{\braket{\bar{v}_{ii}|\bar{v}_{ii}}}=\beta_{E}+\alpha_{E}.
    \end{equation}
    Consequently, we have 
    \begin{equation}
    \begin{split}
        \|\mathcal{P}_{E}(\ket{\psi}\bra{\psi})\|_{\rm tr} & \leq \|\mathcal{P}_{E}(\sigma)\|_{\rm tr}+\|\mathcal{P}_{E}(\ket{\psi}\bra{\psi})-\mathcal{P}_{E}(\sigma)\|_{\rm tr}\\
        & \leq \beta_{E}+\alpha_{E}+2\beta_{E}\sum\limits_{i<j}\lambda_{i}\lambda_{j}.
    \end{split}
    \end{equation}
    Combining this with $C(\ket{\psi})\geq 2\sqrt{\dfrac{2}{d(d-1)}}\sum\limits_{i<j}\lambda_{i}\lambda_{j}$ \cite{LowBound0}, we can finish the proof. \qed

\section{The Hermitian basis operators used to construct $(8,2)$-POVM in the Example} \label{Appendix A}
In the Example, we used an $(8,2)$-POVM constructed from the Hermitian basis operators $G_{\alpha, k}$ given by the following Gell-Mann matrices\cite{GellMann},
    $$G_{11}=\dfrac{1}{\sqrt{2}}\begin{pmatrix}
    	0 & 1 & 0\\
    	1 & 0 & 0\\
    	0 & 0 & 0
    \end{pmatrix},~~~
    G_{21}=\dfrac{1}{\sqrt{2}}\begin{pmatrix}
    	0 & -\mathrm{i} & 0\\
    	\mathrm{i} & 0 & 0\\
    	0 & 0 & 0
    \end{pmatrix},~~~
    G_{31}=\dfrac{1}{\sqrt{2}}\begin{pmatrix}
    	0 & 0 & 1\\
    	0 & 0 & 0\\
    	1 & 0 & 0
    \end{pmatrix},
$$
$$G_{41}=\dfrac{1}{\sqrt{2}}\begin{pmatrix}
    	0 & 0 & -\mathrm{i}\\
    	0 & 0 & 0\\
    	\mathrm{i} & 0 & 0
    \end{pmatrix},~~~
    G_{51}=\dfrac{1}{\sqrt{2}}\begin{pmatrix}
    	0 & 0 & 0\\
    	0 & 0 & 1\\
    	0 & 1 & 0
    \end{pmatrix},~~~
    G_{61}=\dfrac{1}{\sqrt{2}}\begin{pmatrix}
    	0 & 0 & 0\\
    	0 & 0 & -\mathrm{i}\\
    	0 & \mathrm{i} & 0
    \end{pmatrix},
    $$
    $$G_{71}=\dfrac{1}{\sqrt{2}}\begin{pmatrix}
    	1 & 0 & 0\\
    	0 & -1 & 0\\
    	0 & 0 & 0
    \end{pmatrix},~~~
    G_{81}=\dfrac{1}{\sqrt{6}}\begin{pmatrix}
    	1 & 0 & 0\\
    	0 & 1 & 0\\
    	0 & 0 & -2
    \end{pmatrix}.
    $$
\end{appendices}
\end{document}